\documentclass[aps,prl,notitlepage,twocolumn,superscriptaddress,longbibliography,nofootinbib]{revtex4-1}
\pdfoutput=1
\usepackage[utf8]{inputenc}

\usepackage{amsmath,amsthm,amssymb,amsfonts}
\usepackage{scalerel}
\usepackage{mathtools}
\usepackage{graphicx}
\usepackage{subfigure}
\usepackage{bbm}
\usepackage[normalem]{ulem}
\usepackage[colorlinks = true,
            linkcolor = red,
            urlcolor  = magenta,
            citecolor = red,
            anchorcolor = blue]{hyperref}
\usepackage{mathrsfs}
\usepackage{bbold}
\usepackage{units}
\allowdisplaybreaks
\usepackage{bm}
\usepackage{braket}
\usepackage{makecell}
\usepackage[nameinlink]{cleveref} 
\usepackage{upgreek}
\usepackage{blindtext}
\usepackage{verbatim}
\usepackage{algorithm}
\usepackage[noend]{algpseudocode}
\usepackage[dvipsnames]{xcolor}
\usepackage{bbm}
\usepackage{array}
\usepackage{multirow}
\usepackage{tabularx}
\usepackage{float}
\usepackage{dcolumn}
\usepackage{slashed}
\usepackage{braket}
\usepackage{verbatim}
\usepackage{multirow}
\usepackage{resizegather}
\usepackage{titlesec}
\usepackage[toc,page]{appendix}
\usepackage[export]{adjustbox}
\usepackage{tikz}
\usetikzlibrary{shapes.geometric}

\newcommand{\sz}{\sigma^{z}}
\newcommand{\sx}{\sigma^{x}}
\newcommand{\tr}{\mathrm{tr}}
\theoremstyle{definition}
\theoremstyle{plain}
\newtheorem{thm}{Theorem}
\newtheorem{thm*}{Theorem}

\newtheorem*{prop*}{Property}
\newtheorem{claim}{Claim}

\newtheorem*{ctr*}{Conjecture}
\newtheorem{lemma}{Lemma}
\newtheorem*{lemma*}{Lemma}
\newtheorem*{replemma}{Lemma 2}
\newtheorem{coro}{Corollary}
\newtheorem{rcoro}{Corollary}
\newtheorem*{coro*}{Corollary}
\begin{document}

\title{$p$-Body $\simeq$ Range $p-1$: Exact Order–Range Mapping and Dual-Unitarity \\
}
\author{Tanay Pathak\,\,\href{https://orcid.org/0000-0003-0419-2583}
{\includegraphics[scale=0.05]{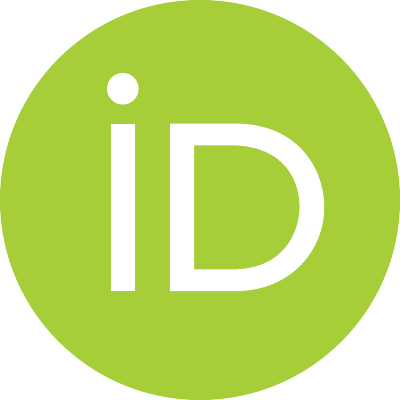}}}
\email{pathak.tanay.4s@kyoto-u.ac.jp}
\affiliation{Department of Physics, Kyoto University, Kitashirakawa Oiwakecho, Sakyo-ku, Kyoto 606-8502, Japan}

\begin{abstract}
We introduce and study a periodically kicked version of a one-dimensional spin-$1/2$ Ising model with contiguous $p$-body interactions. At specific interaction strengths of the model, we establish an exact Floquet interaction \emph{order–range} mapping:  the Floquet evolution operator of the $p$-body Ising model is exactly mapped to a two body Ising model with interactions having maximum range $p-1$ and exactly determined couplings. Upon further tuning the kicking strength, the mapped model becomes dual-unitary,
thus providing an explicit family of $p$-body dual unitary models.
We then determine exact entanglement dynamics of the $p$-body model for a class of solvable initial states at all times. We further illustrate the nontrivial structures that can emerge from the \emph{order–range} mapping using a minimal example of $p=3$ and generalize it for all odd $p$. These findings provide a systematic route to constructing and studying $p$-body dual-unitary Floquet models.
\end{abstract}

\maketitle
%{\huge We can put some of the information in the Supplemental Materials \cite{supp}}
%~~~~~~~~~~~~~~~~~~~~~~~~~~~~~~~~~~~~~~~~~~~~~~~~~~~~~~~~~~~~~~~~
\emph{Introduction.} Understanding the non-equilibrium dynamics of interacting quantum many-body systems remains a central challenge. The exponential growth of the Hilbert-space dimension with system size makes exact calculations, and often even controlled numerical simulations of long-time dynamics, prohibitively difficult. Analytically tractable models that retain genuine interactions are therefore especially valuable. A particularly important class of such models is based on space-time symmetry, a discrete symmetry between spatial and temporal directions to study the quantum dynamics \cite{PhysRevLett.102.240603,Akila_2016,Bertini:2025ddr}. This property called \emph{dual unitarity} has been extensively applied to understand a range of problems in quantum many-body physics, such as the exact dynamics of correlation functions \cite{PhysRevLett.123.210601,PhysRevB.102.174307,PhysRevLett.126.100603,PhysRevB.101.094304,PhysRevLett.133.170403}, studies of entanglement spreading \cite{Bertini:2018fbz,PhysRevLett.132.120402,Gopalakrishnan:2019pip,Bertini:2019gbu,Bertini:2019wkb,PhysRevLett.125.070501,Reid:2021fsg,Zhou:2022uuw,PhysRevB.107.174311,PRXQuantum.6.010324,Pathak:2026rfk}, and the spectral form factor \cite{Bertini:2018wlu,Bertini:2020mdd,PhysRevResearch.2.043403,PhysRevX.11.021051,PhysRevResearch.6.033226,PhysRevB.111.094316}, among many others. Recent extensions and generalizations of dual-unitarity have also been developed in several directions, including hierarchical/generalized dual-unitary circuits \cite{Yu:2023tmz,Rampp:2026cdu} and generalized space-time duality \cite{Jonay:2021kgl,PhysRevResearch.6.033271,PhysRevLett.132.250402,PhysRevResearch.7.L012011,Rampp:2024apt}. Within this framework the kicked-field Ising model (KFIM) \cite{Prosenkfim_1998,prosenkfim,PhysRevE.65.036208,Prosen:2007hwp,PhysRevE.76.061127} presents itself as a particularly useful model. Owing to its dual-unitarity \cite{Akila_2016}, the model often serves as a minimal setting to study quantum chaos and non-equilibrium dynamics in interacting quantum many-body systems. Much of the theoretical framework of these models however, is built mostly around two-body interactions. Extensions beyond this setting include interaction round-a-face circuits \cite{Prosen:2021wvx,Prosen:2021hks,Claeys:2023yiv} and tri-unitary circuits \cite{Jonay:2021kgl,PhysRevResearch.7.L012011}. Genuine multi-body interactions also arise naturally in various other situations as well such as in lattice gauge theories \cite{Kitaev_2003,Kitaev:2005hzj,Chandrasekharan_1997}, effective spin Hamiltonians with ring-exchange processes \cite{PhysRevB.37.9753,PhysRevB.39.2299,PhysRevLett.86.5377}, stabilizer and cluster-state models of quantum information \cite{PhysRevLett.93.056402,PhysRevA.71.062313,PhysRevA.69.062311,Hein_2004}. They have also been simulated in quantum computers \cite{Dai:2016coe, Katz_2023,Satzinger_2021}. 

Motivated by these, we take a step further and investigate the possibility of dual unitarity in a model with $p$-body interactions. To achieve this we first introduce a periodically kicked version of the $p$- body Ising model \cite{turban1982self,PhysRevB.29.519,PhysRevB.26.6334,kolb1986conformal,alcaraz1987conformal,igloi1986series,PhysRevB.108.214430} (also see  \cite{note1}). We then establish an exact mapping between the Floquet evolution operators of the kicked $p-$ body Ising model and the Floquet evolution operator of the kicked Ising model consisting of only $2$- body interactions with maximum range $p-1$. This mapping, as we later show, is independent of the dual unitarity of the model, and only dependent on the interaction strength. Upon further tuning the kicking strength of the $p$- body model we observe that the resulting mapped model becomes dual unitary, thus providing a concrete first realization of its kind. This further allows for exact evaluation of entanglement dynamics of a class of initial solvable states \cite{Pathak:2026jfn}. We also analyze the minimal case of $p=3$, highlighting the rich structure of these models.

%~~~~~~~~~~~~~~~~~~~~~~~~~~~~~~~~~~~~~~~~~~~~~~~~~~~~~~~~~~~~~~~~
\emph{The model.} Our starting point is the $p$-body generalization of a one dimensional spin-$1/2$ Ising model \cite{turban1982self,PhysRevB.29.519,PhysRevB.26.6334,kolb1986conformal,alcaraz1987conformal,igloi1986series,PhysRevB.108.214430} of length $L$. The Hamiltonian of the model is given by 
\begin{equation}
    H^{(pB)} = \sum_{j=1}^{L} J \sigma^{z}_{j}\sigma^{z}_{j+1}\cdots \sigma^{z}_{j+p-1} + \sum_{j=1}^{L} b\, \sigma^{x}_{j}.
\end{equation}
Here, $pB\equiv \text{p-body}$ and  $\sz_{j},\sx_{j}$ are the Pauli matrices acting non-trivially on the $j-$th site and we assume periodic boundary conditions as $\sz_{L+1} \equiv \sz_{1}$. Let us now introduce a kicked version of this model \emph{akin} to the standard KFIM \cite{prosenkfim,PhysRevE.65.036208,Prosen:2007hwp}. Setting time interval between two kicks to be unity we have the following Hamiltonian  
\begin{align}\label{eq:pbodyham}
    H^{(pB)}_{\text{KFIM}}&= H_{I} + H_{K} \sum_{\tau= -\infty}^{\infty} \delta(t- \tau), \\
    H^{(pB)}_{I}& = \sum_{j=1}^{L} J \sigma^{z}_{j}\sigma^{z}_{j+1}\cdots \sigma^{z}_{j+p-1} + \sum_{j=1}^{L}h_{j}\sigma_{j}^{z}, \\
    H^{(pB)}_{K}& = \sum_{j=1}^{L} b\, \sigma^{x}_{j}.
\end{align}
The total Floquet operator of the system is 
\begin{equation}\label{eq:pbodyu}
    U^{(pB)}_{\text{KFIM}}\equiv U^{(pB)}_{\text{KFIM}}[J,b,\mathbf{h}] = U_{K}U_{I}^{(pB)},
\end{equation}
where $U_{I}^{(pB)}= e^{-i H_{I}^{(pB)}}$ and $U_{K}= e^{-i H_{K}}$ and  $J,b$ and $\mathbf{h}= \{h_{1},h_{2},\cdots, h_{L}\}$ are generic. 

% \begin{figure}[H]
%     \centering
%     \includegraphics[width=  \linewidth]{order_range.pdf}
% \caption{Schematic of interaction \emph{order-range} mapping for minimal case of $3$-body kicked Ising model and interaction strength $\frac{\pi}{4}$. In the Floquet setting the 3-body interaction model is mapped to another model with 2-body interactions with maximal range of 2. The mapping infact depends on the length of chain, $L$. For even $L$ the model is mapped to two independent decoupled chains, while for odd $L$ the model becomes a standard KFIM after some relabeling, as we later show.}\label{fig:schematic}
% \end{figure}

%~~~~~~~~~~~~~~~~~~~~~~~~~~~~~~~~~~~~~~~~~~~~~~~~~~~~~~~~~~~~~~~~
\begin{figure}[t]
    \centering
    \includegraphics[width=0.5\textwidth]{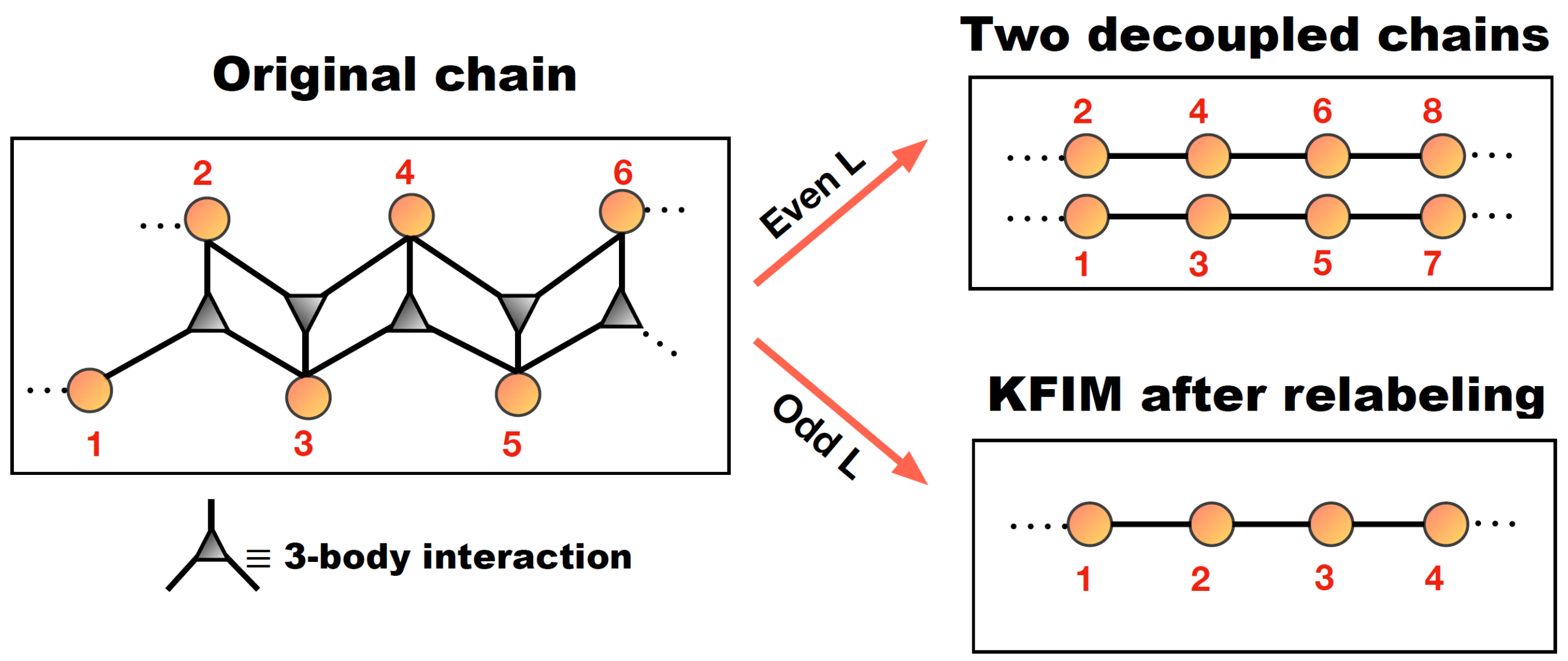}
    \caption{Schematic of interaction \emph{order-range} mapping for minimal case of $3$-body kicked Ising model and interaction strength $\frac{\pi}{4}$. In the Floquet setting the 3-body interaction model is mapped to another model with 2-body interactions with maximal (possible) range of 2. The mapping further depends on the length of chain, $L$. For even $L$ the model maps onto two independent, decoupled chains, while for odd $L$ it is equivalent to KFIM after relabeling, as we later show.}
    \label{fig:orderrange}
\end{figure}
\emph{Exact order-range mapping.} At present the model does not offer any further simplifications. Aiming to further simplify the model we first note the following lemma, that will be useful.
\begin{lemma}\label{lemma1}
    For $J= \frac{n \pi}{2^{p-1}}$, $\forall{n} \in \mathbb{Z}$, following relation holds:
    \begin{equation}
        e^{-iJ (-1)^{p+1}\prod_{k=1}^{p}(1-\sigma^{z}_{i_k})}= \mathbbm{I},
    \end{equation}
where $i_{1} < i_{2} < \cdots <i_{p}$ and $\mathbbm{I}$ denotes the identity matrix. 
\end{lemma}
%~~~~~~~~~~~~~~~~~~~~~~~~~~~~~~~~
\emph{Proof.} Let $\ket{s}$ denote the eigenstate of $\sigma^{z}$ such that $\sigma_{z}\ket{s}= s \ket{s}$; $s\in \{1,-1\}$. $\psi_{s}=\ket{s_{1},s_{2},\cdots,s_{L}}$ then denote a basis state for chain of length $L$. Acting on $\psi_{s}$, the LHS has following two possibilities
\begin{enumerate}
    \item If we have $\ket{s_{i_k}}=\ket{+1}$ for any $ k \in\{1,\cdots,p\}$, then $e^{-iJ (-1)^{p+1}\prod_{k=1}^{p}(1-s_{i_k})}= 1$ and the Lemma holds trivially.
    \item The only remaining case is when $\ket{s_{i_k}}= \ket{-1}, \forall k \in\{1,\cdots,p)$. To ensure that LHS is equal to 1 we require: $e^{-i2^{p}(-1)^{p+1}J}= 1 \implies J= \frac{n \pi}{2^{p-1}}.$
\end{enumerate}
This concludes the proof. It is also possible to exactly determine the expansion of $ \prod_{k=1}^{p}(1-\sz_{i_{k}})$ appearing in the left hand side of Lemma \eqref{lemma1}. It is explicitly given using the following Lemma.
\begin{lemma}\label{lemma2}
    Let $\sz_{i}$, denote the $\sz$ operator acting non-trivially on the $i$-th site of length $L$ chain, then the following is true.
    \begin{align}
    \prod_{k=1}^{p}(1-\sz_{i_{k}})&= 1 - \sum_{k=1}^{p}\sz_{i_{k}} + \sum_{1 \leq k_{1}<k_{2}\leq p }^{p}\sz_{i_{k_{1}}}\sz_{i_{k_{2}}}- \cdots \nonumber\\
    & + (-1)^{p}\sz_{i_{1}}\sz_{i_{2}}\cdots \sz_{i_{p}}. 
\end{align}
\end{lemma}
See \cite{supp} for a proof. 
%~~~~~~~~~~~~~~~~~~~~~~~~~~~~~~~
A direct consequence of Lemma \eqref{lemma1} and \eqref{lemma2} is that for the specific values of the interaction strength, i.e. $J= \frac{n \pi}{2^{p-1}}$; $\forall{n} \in \mathbb{Z}$, one can reduce the Floquet operator of the $p$- body model to the Floquet operator of the $(p-1)$-body model. This already offers a simplification as it means that the dynamics of the $p$-body model for specific interaction strengths is exactly equal to the dynamics of the $(p-1)-$ body model. However, a further reduction is not possible as the resulting expression still contains terms of all interaction order $\leq p-1$ and the coupling is a function of $p$.  We however note that the strongest possible constraint is obtained from $p=3$, i.e. if we set $J= \frac{n \pi}{4};n\in\mathbb{Z}$, then we can recursively reduce, using Lemma \eqref{lemma1} and \eqref{lemma2}, each term with $p >2$ body interaction completely, to a $2$-body interaction term  having maximal range of $p-1$. Thus we obtain the following order–range mapping theorem, that constitutes a main result of this work.
%~~~~~~~~~~~~~~~~~~~~~~~~~~~~~~~~
\begin{thm}[Order--Range mapping]
    For a $p$-body Ising model of size $L$, with total Floquet evolution operator $U^{(pB)}_{\text{KFIM}}$, as given by Eq. \eqref{eq:pbodyu}, and with $J= \frac{n \pi}{4};n\in\mathbb{Z}$, we have the following exact relation for the Floquet evolution operators
    \begin{equation}
       U^{(pB)}_{\text{KFIM}}= \exp[i \phi(L)] U^{((p-1)R)},
    \end{equation}
where $(p-1)R\equiv \text{range }p-1$, $\phi(L)$ is a system size dependent global phase,  and $U^{((p-1)R)}$ is the Floquet evolution operator of a spin-$1/2$ model with $2$-body interactions of maximum range $p-1$.
\end{thm}
See \cite{supp} for a proof. We have thus obtained a huge simplification as a result of interaction \emph{order-range} mapping i.e. the Floquet evolution operator of the $p$-body Ising model coincides, up to a global phase, to an Ising model with only $2$- body interactions having maximum range of $p-1$. Remarkably, it is also possible to obtain the Hamiltonian corresponding to the mapped model explicitly, which is given as follows (see \cite{supp})
\begin{align}\label{eq:3rmodel}
H^{((p-1)R)}_{\text{KFIM}}&= H_{I} + H_{K} \sum_{\tau= -\infty}^{\infty} \delta(t- \tau)
\end{align}
 Here, $H^{((p-1)R)}_{I} = J \sum_{r=1}^{p-1}(p-r) \sum_{i=1}^{L}\sz_{i}\sz_{i+r} +\sum_{i} \tilde{h}_{i} \sigma_{i}^{z}$ and $H_{K} = b \sum_{i=1}^{L} \, \sigma^{x}_{i}$, $J= \frac{n\pi}{4}; n\in \mathbb{Z}$ and $\tilde{h}_{i}=h_{i} -p(p-2)J$. The mapping obtained so far does not require any special value of the kicking strength $b$ and is valid for its generic values. As shown in Ref. \cite{Pathak:2026jfn,Osipov:2026xyr} a range $r$ kicked Ising chain is dual unitary when the maximal range ($\equiv r)$ couplings equals $|\frac{\pi}{4}|$ and $|b|= \frac{\pi}{4}$, while the shorter range couplings can remain arbitrary. Since the original $p$-body Ising model and mapped range $p-1$ Ising model, Floquet operators coincide up to a global phase, the original $p$-body model has dual unitarity for $|J|= \frac{\pi}{4}$ and $|b|= \frac{\pi}{4}$. Thus, we have accomplished the task of obtaining models containing contiguous $p$- body interactions, which are \emph{dual unitary}.  
%~~~~~~~~~~~~~~~~~~~~~~~~~~~~~~~~~~~~~~~~~~~~~

\emph{Three-body model.} We now consider a minimal case of $p=3$, for which we have the following Hamiltonian 
\begin{align}
    H_{I}^{(3B)}& = \sum_{i=1}^{L} J \sigma^{z}_{i}\sigma^{z}_{i+1} \sigma^{z}_{i+2} + \sum_{i}h_{i}\sigma_{i}^{z},\nonumber \\
    H_{K}& = \sum_{i=1}^{L} b\, \sigma^{x}_{i},
\end{align}
 with periodic boundary conditions. We first study the spacing distribution of quasi energies of the model. Let $e^{i \theta_{i}}$ denote the $i$-th eigenvalue of the unitary Floquet operator $U$, then $\theta_{i}$ are the quasi energies. Arranging the quasi energy in ascending order we obtain the difference between consecutive quasi energies i.e. $s_{i}= \theta_{i+1}-\theta_{i}$. We are interested in the distribution of the unfolded spacings i.e. $\tilde{s}= \frac{\mathcal{D}}{2\pi}s$ where $\mathcal{D} =2^{L}$ denotes the dimension of the Hilbert space of $U$. In Fig. \eqref{fig:levelspec3} we show the numerical results for the spacing distribution for the $3-$body model. We choose $L=13,14$ and consider 128 and 64 disorder realizations respectively, $J=\frac{\pi}{4}, b = -\frac{\pi}{4}$ and $h_{i}$ are chosen to be Gaussian random variable with mean $0$ and variance $\frac{\pi}{4}+1$. 

\begin{figure}[H]
    \centering
    \includegraphics[width=\linewidth]{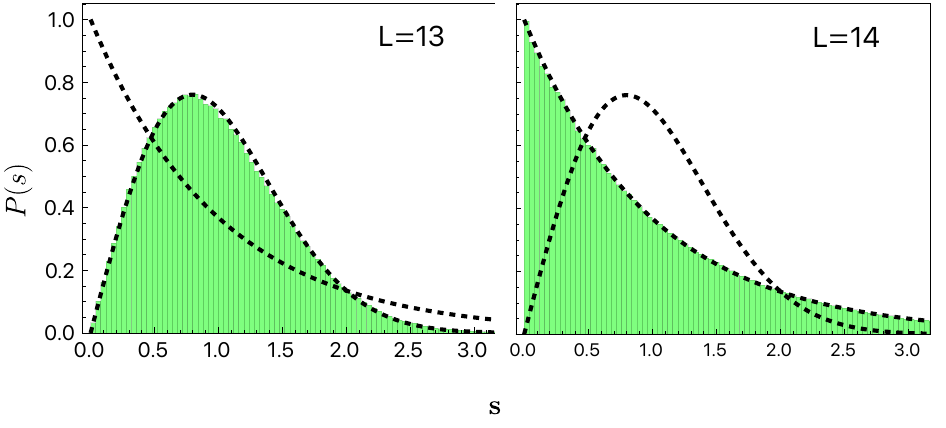}
    \caption{ Level spacing distribution of the (unfolded) quasi energies of the unitary operator for $3$-body case (a) $L =13$ (b) $L= 14$.  The black dashed curves correspond to the Wigner surmise for COE (= $\frac{1}{2} \pi  s \exp \left(\frac{1}{4} (-\pi ) s^2\right)$) and Poisson distribution (=$e^{-s}$).}
    \label{fig:levelspec3}
\end{figure}

The quasienergy spacing distribution observed is Poisson, usually associated with integrable dynamics, for even $L$, whereas it follows the Wigner-Dyson distribution corresponding to circular orthogonal ensemble (COE), usually associated with chaotic dynamics, for odd $L$. At first glance the behavior looks mysterious, however it becomes clear once we analyze the mapped range $2$ model in detail. For a chain of length $L$ the Hamiltonian of the mapped model is 
\begin{align}
    H_{I}^{(2R)}& = \frac{\pi}{4}\sum_{i=1}^{L} \sigma^{z}_{i}\sigma^{z}_{i+2}+  \sum_{i}\left(h_{i}-\frac{3\pi}{4}\right)\sigma_{i}^{z},\nonumber \\
    H_{K}& = -\frac{\pi}{4}\sum_{i=1}^{L} \, \sigma^{x}_{i}.
\end{align}
For even $L$ notice that the $H_{I}^{(2R)}$ is written as sum of the contribution from only odd and only even site 
\begin{align}
   H_{I}^{(2R)}= \frac{\pi}{4}\sum_{i\in \text{odd}} \sigma^{z}_{i}\sigma^{z}_{i+2}+  \sum_{i\in \text{odd}}\left(h_{i}-\frac{3\pi}{4}\right)\sigma_{i}^{z} + \nonumber \\
   \frac{\pi}{4}\sum_{i\in \text{even}} \sigma^{z}_{i}\sigma^{z}_{i+2}+  \sum_{i\in \text{even}}\left(h_{i}-\frac{3\pi}{4}\right)\sigma_{i}^{z}.
\end{align}
Thus, the model is equivalent to two decoupled dual unitary chains of the same form. Notice that with this mapping the interactions are nearest-neighbor within each chain. The total Floquet operator can then be decomposed as $U^{(2R)}= U_{\text{odd}}\otimes U_{\text{even}}$. If $\lambda_{i}$ and $\mu_{i}$ denote the eigenvalues of $U_{\text{odd}}$ and $U_{\text{even}}$ respectively, then the eigenvalues of $U^{(2R)}$ are given by $\lambda_{i}\mu_{j}$; $i,j=1,2,\cdots,2^{L/2}$. Thus, although each of $U_{\text{odd}}$ and $U_{\text{even}}$ can be chaotic, the total $U^{(2R)}$ can have uncorrelated spectra, which we observe. For odd $L$ however, due to the chosen periodic boundary conditions we obtain the following mapped Hamiltonian
\begin{align}
   H_{I}^{(2R)}&= \frac{\pi}{4}\sum_{i\in \text{odd}} \sigma^{z}_{i}\sigma^{z}_{i+2}+  \sum_{i\in \text{odd}}\left(h_{i}-\frac{3\pi}{4}\right)\sigma_{i}^{z} + \nonumber \\
   &\frac{\pi}{4}\sum_{i\in \text{even}} \sigma^{z}_{i}\sigma^{z}_{i+2}+  \sum_{i\in \text{even}}\left(h_{i}-\frac{3\pi}{4}\right)\sigma_{i}^{z} \nonumber \\
   &+ \frac{\pi}{4} \sigma^{z}_{L}\sigma^{z}_{2} + \frac{\pi}{4}\sigma^{z}_{L-1}\sigma^{z}_{1},\nonumber \\
    H_{I}^{(2R)}&= \frac{\pi}{4}\sum_{i=1}^{L} \sigma^{z}_{i}\sigma^{z}_{i+2}+  \sum_{i=1}^{L}\left(h_{i}-\frac{3\pi}{4}\right)\sigma_{i}^{z}.
\end{align}
In the last line we do the relabeling as : $ 2n - 1 \mapsto n \, \text{for odd sites and }
  2n \mapsto \frac{L+1}{2} +n \, \text{for even sites.} 
$
Thus we obtain a remarkable relation that the 3-body KFIM with interaction strength $|J|= \frac{\pi}{4}$ is equivalent to the standard KFIM. This explains the chaotic spectra for odd $L$. These results thus highlight the rich structure underlying the $p$-body models. This connectivity structure extends to all odd $p$ as given by the following corollary
\begin{coro}\label{corr1}
    For odd $p$, with periodic boundary conditions, $J= \frac{(2n+1) \pi}{4}$, we have the following two cases:
    \begin{itemize}
    \item For even $L$ the model is mapped, up to a global phase, to two decoupled chains each of which has maximum range $\frac{p-1}{2}$ interactions. 
        \item For odd $L$ the model is mapped, up to a global phase, to a single chain with maximum range $\frac{p-1}{2}$ interactions. 
    \end{itemize}
\end{coro}
See \cite{supp} for a proof. This further implies that for any odd $p$ and even $L$ the level statistics will be Poisson. 
%~~~~~~~~~~~~~~~~~~~~~~~~~~~~~~~~~~~~~~~~~~~~~

\emph{Exact entanglement dynamics.} 
We next consider the exact entanglement dynamics of the $p$-body model. The initial state chosen is a product state, as follows:
\begin{equation}
    \ket{\psi_{\theta,\phi}}= \bigotimes_{k=1}^{L} \,\left(\cos\left(\frac{\theta_{k}}{2}\right) \ket{\uparrow} + e^{i\phi_{k}} \sin\left(\frac{\theta_{k}}{2}\right) \ket{\downarrow}\right),
\end{equation}
where $\ket{\uparrow}$ and $\ket{\downarrow}$ denote the eigenstates of $\sz$, $\theta_{k}\in [0,\pi]$ and $\phi_{k}\in [0,2\pi] $. We specifically consider the two classes of states: transverse ($\mathcal{T}$) and longitudinal ($\mathcal{L}$), which are specified by the following parameters
\begin{align}
    &\mathcal{T}= \{\ket{\psi_{\pi/2,\phi}}, \forall\, \theta_{k}= \frac{\pi}{2}, \phi_{k} \in [0,2\pi]\},  \\
    &\mathcal{L}= \{\ket{\psi_{\bar{\theta},\phi}}, \bar{\theta}=\{ 0,\pi\}, \phi_{k} \in [0,2\pi]\}.
\end{align}
These states are called solvable states as they allow for exact determination of the entanglement dynamics \cite{Bertini:2018fbz,Pathak:2026jfn}. Other states which do not belong to these solvable classes will be called generic states.
To quantify the entanglement dynamics of the system we evaluate the $\alpha$-R\'enyi entropies which are given as follows
\begin{equation}
    S^{(\alpha)}_{A}(t)= \frac{1}{1-\alpha} \ln\tr[(\rho_{A}(t))^{\alpha}]\quad \alpha > 0.
\end{equation}
where $\rho_{A}(t)$ is the time evolved reduced density matrix corresponding to a contiguous set of $N$ spin; $A= \{1,2,\cdots,N\}$. Since the $p$- body model is mapped to range $p-1$ model which is dual unitary, the $\alpha$-R\'enyi entropy for all times and initial solvable states of $\mathcal{T}$-type, is given as \cite{Pathak:2026jfn}
\begin{align}\label{eq:eefomula}
   \lim_{L\rightarrow\infty} S^{(\alpha)}_{A}(t)= \min(2(p-1)t,N) \ln(2).
\end{align}
\begin{figure}[!t]
    \centering
    \includegraphics[width= 0.95\linewidth]{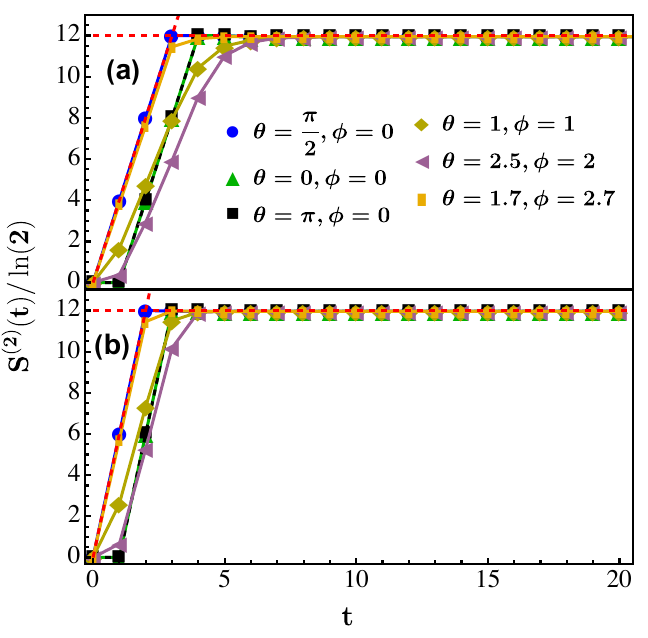}
    \caption{Time evolution of $2$-R\'enyi entanglement entropy (in units of $\ln(2)$) for (a) $p= 3$ model with $h_{i}= \frac{\pi}{4}$ (b) $p= 4$ model with $h_{i}=1$. In both cases the total system size is $L= 29$ and subsystem size is $N=12$. The initial states are of $\mathcal{T}$-type, $\mathcal{L}$-type and generic type with their parameter values given in the legend. The inclined dashed lines indicate the predicted linear growth of $4t$ in (a)  and $6t$ in (b), while the horizontal dashed lines correspond to the saturation value of $N$.}
    \label{fig:eekfimp34}
\end{figure}
The results for $\mathcal{L}$-type states are same as that for $\mathcal{T}$-type states but delayed by a period.
In Fig. \eqref{fig:eekfimp34} we numerically study the spread of entanglement in $p=3,4$ body models. The initial states considered belong to $\mathcal{T}$ class, $\mathcal{L}$ class and some generic states. The total system size is $L=29$, sub-system size $N=12$ and $J = \frac{\pi}{4}, b= - \frac{\pi}{4}$. For the states belonging to $\mathcal{T}$ or $\mathcal{L}$ classes the growth regime and saturation for both $p=3,4$ are in agreement with Eq. \eqref{eq:eefomula} (delayed by a period for $\mathcal{L}$-type states). For generic states the observed slope deviates from $2(p-1)$, while the saturation values equal $N$, in the units of $\ln(2)$. These differences in slopes are however expected to go away in the limit of $N \rightarrow \infty$ \cite{Pathak:2026jfn}. It is however important to note one subtlety for the case of $p=3$, where we previously have shown that the model is equivalent to the standard KFIM. However, the entanglement entropy growth observed is $4t\ln(2)$. The resolution to this puzzle is in the relabeling of the sites. The numerical calculations were performed using $3-$body operators, thus the sites $1,2,3,\cdots,L$ are contiguous while after the relabeling these sites are not contiguous. Thus in the standard KFIM picture the entanglement entropy is obtained for bi-partitions which are disjoint, thus the entropy calculation requires twice the number of cuts, which explains the observed behavior.  
%~~~~~~~~~~~~~~~~~~~~~~~~~~~~~~~~~~~~~~~~~~~~~~~~

\emph{Discussion and Outlook.} We introduced a periodically kicked version of a one-dimensional spin-$1/2$ Ising model with contiguous $p$-body interactions.. For specific values of interaction strengths we established an exact interaction \emph{order-range} duality for this model which maps the $p$-body model to a model which consists of only 2-body interactions having maximum range $p-1$. Furthermore, for special values of the kicking strength the mapped long range model becomes dual unitary as well, thus making the original $p$-body model dual unitary. This is a first concrete realization of a contiguous $p$-body kicked-Ising Floquet model which is also dual unitary. We then take the minimal case of $p=3$ and total size $L$ and show that the mapped model is either two disconnected chains or a single connected chain depending on whether $L$ is even or odd respectively. It is also shown that for odd $L$ the model is exactly mapped to the standard KFIM, after proper labeling of the sites. Owing to the dual unitarity of the model we obtain exact expressions for all $\alpha$- R\'enyi entropies for all times and verify the results with numerical simulations. Our work thus rigorously established a possible way of constructing $p$-body dual unitary Floquet systems.

Being first such construction of $p$-body dual unitary model, our work  motivates further studies. An important question is to understand the scenarios where the interaction order-range mapping is valid. Although explicitly presented for spin-$1/2$ $p$-body KFIM, we expect it to hold for a $p$-body generalization of the kicked Potts model for $q=3,4$ \cite{PhysRevB.105.144306,Claeys:2024tuy}. This will then naturally lead to the construction of $p$- body dual unitary models with local dimensions$=3,4$. It is also important to understand the microscopic origin of the interaction order-range mapping. A plausible way would be to consider the $p$- local gate corresponding to the $p$-body model and understand its decomposition into range $p-1$ gates. This might be also related to the decomposition of the diagonal gates which have been explored well in recent times \cite{deBeaudrap:2020oef, Cui:2016bxt,Sriluckshmy:2023leq}. This will pave the way to construct generalized circuit models with $p$-body interactions which are dual unitary. Finally we remark that the interaction order-range mapping is valid independently of the dual unitarity and it remains to find its use in more generic situations.

\emph{Acknowledgments.}We thank Toma\v{z} Prosen for his useful comments. We also thank Abhik for his suggestions to improve the presentation of the manuscript. The author is supported by JST CREST (Grant No. JPMJCR24I2). Numerical calculations have been performed using the computational facilities of the Yukawa Institute for Theoretical Physics. The author also acknowledges use of GPT 5.5 for suggesting a very special case of Lemma \eqref{lemma1} which then led to its current generalized form.
\bibliography{references}

%~~~~~~~~~~~~~~~~~~~~~~~~~~~~~~~~~~~~~~~~~~~~~~~~
% \clearpage
\onecolumngrid
% \hbox{}\thispagestyle{empty}\newpage

\onecolumngrid
\vspace{1cm}
\raggedbottom
\begin{center}
\textbf{\large Supplemental Material for:\\
$p$-Body $\simeq$ Range $p-1$: Exact Order–Range Mapping and Dual-Unitarity}
\vspace{2ex}

{\normalsize Tanay Pathak$^{1}$ \,\,\href{https://orcid.org/0000-0003-0419-2583}
{\includegraphics[scale=0.05]{orcidid.pdf}}}\\
\small $^{1}$Department of Physics, Kyoto University, Kitashirakawa Oiwakecho, Sakyo-ku, Kyoto 606-8502, Japan

\end{center}
\appendix
In this supplemental material, we provide proofs of Lemma 2 and Theorem 1  of the main text. In particular:
\begin{itemize}
\item In Section \ref{ssec:lemma2proof} we prove Lemma 2 of the main text. 
\item In Section \ref{ssec:thm1proof} we prove Theorem 1 of the main text. 
\item In Section \ref{ssec:cor1proof} we prove Corollary 1 of the main text.
\end{itemize}

\section{Proof of Lemma 2}\label{ssec:lemma2proof}
\begin{replemma}\label{suppsec:lemma2}
Let $\sz_{i}$, denote the $\sz$ operator acting non-trivially on the $i$-th site of length $L$ chain, then following holds.
    \begin{equation}\label{suppsec:claim1}
    \prod_{k=1}^{p}(1-\sz_{i_{k}})= 1 - \sum_{k=1}^{p}\sz_{i_{k}} + \sum_{1 \leq k_{1}<k_{2}\leq p }^{p}\sz_{i_{k_{1}}}\sz_{i_{k_{2}}}- \cdots + (-1)^{p}\sz_{i_{1}}\sz_{i_{2}}\cdots \sz_{i_{p}}. 
\end{equation}
\end{replemma}
\emph{Proof.} We proceed via induction. Let $P(n)$ be the given statement 
$$P(n): \prod_{k=1}^{n}(1-\sz_{i_{k}})= 1 - \sum_{k=1}^{n}\sz_{i_{k}} + \sum_{1 \leq k_{1}<k_{2}\leq n }\sz_{i_{k_{1}}}\sz_{i_{k_{2}}}- \cdots + (-1)^{n} \sz_{i_{1}}\sz_{i_{2}}\cdots \sz_{i_{n}}.$$
For $n=1$: LHS = $1-\sz_{i_{1}}$, hence $P(1)$ is true. We now assume $P(n)$ is true for some positive integer $n$. To conclude the proof we need to show that $P(n+1)$ is true. For $(n+1)$ we have 
\begin{align}
 \prod_{k=1}^{n+1}(1-\sz_{i_{k}}) &= \left( \prod_{k=1}^{n}(1-\sz_{i_{k}})\right)(1-\sz_{i_{n+1}}), \nonumber \\
 &= \left(1 - \sum_{k=1}^{n}\sz_{i_{k}} + \sum_{1 \leq k_{1}<k_{2}\leq n }\sz_{i_{k_{1}}}\sz_{i_{k_{2}}}- \cdots + (-1)^{n}\sz_{i_{1}}\sz_{i_{2}}\cdots \sz_{i_{n}}\right)(1-\sz_{i_{n+1}}), \nonumber \\
 &= \left(1 - \sum_{k=1}^{n}\sz_{i_{k}} + \sum_{1 \leq k_{1}<k_{2}\leq n }\sz_{i_{k_{1}}}\sz_{i_{k_{2}}}- \cdots + (-1)^{n}\sz_{i_{1}}\sz_{i_{2}}\cdots \sz_{i_{n}}\right), \nonumber \\
 &- \left(\sz_{i_{n+1}} - \sum_{k=1}^{n}\sz_{i_{k}}\sz_{i_{n+1}} + \sum_{1 \leq k_{1}<k_{2}\leq n }\sz_{i_{k_{1}}}\sz_{i_{k_{2}}}\sz_{i_{n+1}}- \cdots + (-1)^{n+1}\sz_{i_{1}}\sz_{i_{2}}\cdots \sz_{i_{n}\sz_{i_{n+1}}}\right), \nonumber \\
 &= 1 - \sum_{k=1}^{n+1}\sz_{i_{k}} + \sum_{1 \leq k_{1}<k_{2}\leq n+1 }\sz_{i_{k_{1}}}\sz_{i_{k_{2}}}- \cdots + (-1)^{n+1}\sz_{i_{1}}\sz_{i_{2}}\cdots \sz_{i_{n}}\sz_{i_{n+1}},
\end{align}
which is exactly the result for $(n+1)$. This concludes the proof.
%~~~~~~~~~~~~~~~~~~~~~~~~~~~~~~~~~
\section{Proof of Theorem 1}\label{ssec:thm1proof}
\begin{thm*}[Order--Range mapping]
   For a $p$-body Ising model of size $L$, with total Floquet evolution operator $U^{(pB)}_{\text{KFIM}}$, and with $J= \frac{n \pi}{4};n\in\mathbb{Z}$, we have the following exact relation for the Floquet evolution operators
    \begin{equation}
       U^{(pB)}_{\text{KFIM}}= \exp[i \phi(L)] U^{((p-1)R)},
    \end{equation}
where $\phi(L)$ is a system size dependent global phase and $U^{((p-1)R)}$ is the Floquet evolution operator of a spin-$1/2$ model with $2$-body interactions of maximum range $p-1$.
\end{thm*}
\emph{Proof.} It is possible to prove the above theorem via induction. However, our interest is also to obtain the resulting range $p-1$ model so we prove it by explicitly obtaining the corresponding range $p-1$ model.

We restrict the proof to $J=  \frac{\pi}{4}$. This is sufficient because all operators appearing in the interaction part commute. Thus for any $n\in \mathbb{Z}$, we have
\begin{equation}
 \exp\left[-i \frac{n \pi}{4} \sum_{i=1}^{L}\sz_{i}\sz_{i+{1}}\cdots\sz_{i+{p-1}} \right]= \exp\left[-i \frac{ \pi}{4} \sum_{i=1}^{L}\sz_{i}\sz_{i+{1}}\cdots\sz_{i+{p-1}} \right]^{n}.
\end{equation}
Therefore, the result for $J=n\pi/4$ follows by taking the $n$-th power of the identities established for $J=\pi/4$.
We next consider a few specific values of $p$ to gain insight on the nature of the resulting decomposition.
%~~~~~~~~~~~~~~~~~~~~~~~~~~~~~~~~~~~~~~~~~~~~~~~~~~~~~~~~~~~~~~~~~~
\subsection{Three-body case}\label{ssec:3bodycase}
Let us first consider the first case of $p=3$, for which we have the following 
\begin{align}
    H^{(3B)}_{I}& = \sum_{i=1}^{L} J \sigma^{z}_{i}\sigma^{z}_{i+1} \sigma^{z}_{i+2} + \sum_{i}h_{i}\sigma_{i}^{z},\nonumber \\
    H_{K}& = \sum_{i=1}^{L} b\, \sigma^{x}_{i}.
\end{align}
And the total Floquet operator is 
\begin{equation}\label{seq:pbodyu}
    U^{(3B)}_{\text{KFIM}}= U_{K}U_{I}, \quad pB \equiv \text{p-body}.
\end{equation}
Next, we consider the $U_{I}$ part of the total Floquet operator and use Lemma 1 to obtain the following decomposition 

\begin{equation}\label{seq:decomp3}
  U_{I}= e^{-i\frac{\pi}{4}\sum \sz_{i}\sz_{i+1}\sz_{i+2} -i \sum_{i=1}^{L}h_{i}\sz_{i} }= e^{-i\frac{\pi}{4}\sum \sz_{i}\sz_{i+1}\sz_{i+2}}e^{-i\frac{\pi}{4}\sum_{i=1}^{L}(1-\sz_{i})(1-\sz_{i+1})(1-\sz_{i+2}) -i \sum_{i=1}^{L}h_{i}\sz_{i} }. 
\end{equation}
From identity \eqref{suppsec:claim1} for $p=3$, we obtain the following
\begin{equation}
(1- \sz_{i_1})(1-\sz_{i_2})(1-\sz_{i_3})= 1- (\sz_{i_1}+\sz_{i_2}+\sz_{i_3}) + (\sz_{i_1}\sz_{i_2} +\sz_{i_1}\sz_{i_3}+\sz_{i_2}\sz_{i_3}) - \sz_{i_{1}}\sz_{i_{2}}
\sz_{i_{3}}.
\end{equation}
Using the above we obtain the following for the RHS of Eq. \eqref{seq:decomp3}
\begin{align}
   % e^{-\frac{\pi}{4}\sum \sz_{i}\sz_{i+1}\sz_{i+2}}&e^{-\frac{\pi}{4}\sum_{i=1}^{L}(1-\sz_{i})(1-\sz_{i+1})(1-\sz_{i+2}) + \sum_{i=1}^{L}h_{i}\sz_{i} }
   \text{RHS}
   &= \exp\left[-i\frac{\pi}{4}\sum_{i=1}^{L} ( 1 - (\sz_{i}+\sz_{i+1}+\sz_{i+2})+\sigma_{i}\sigma_{i+1}+\sigma_{i+1}\sigma_{i+2}+\sigma_{i}\sigma_{i+2}) - i\sum_{i=1}^{L}h_{i}\sz_{i}  \right], \nonumber \\
   &= \exp\left[-i\frac{\pi}{4}L+i \frac{3\pi}{4} \sum_{i=1}^{L} \sz_{i}-i\frac{\pi}{2} \sum_{i=1}^{L}\sigma_{i}\sigma_{i+1}-i\frac{\pi}{4} \sum_{i=1}^{L}\sigma_{i}\sigma_{i+2} -i \sum_{i=1}^{L}h_{i}\sz_{i} \right], \nonumber \\
   &= (-i)^{L}\exp\left[-i\frac{\pi}{4}L-i\frac{\pi}{4} \sum_{i=1}^{L}\sigma_{i}\sigma_{i+2} -i \sum_{i=1}^{L}(h_{i}-\frac{3\pi}{4})\sz_{i} \right],
\end{align}
where in the last line we use the property that $e^{-i\frac{\pi}{2}\sum_{i=1}^{L}\sz_{i}\sz_{i+1}}= \prod_{i=1}^{L}e^{-i\frac{\pi}{2}\sz_{i}\sz_{i+1}}= (-i)^{L}$. We have thus shown  that the discrete time evolution of $3-$body model is equivalent to the discrete time evolution of a $2-$body with range $2$ interactions. The model is explicitly given as follows 
\begin{align}\label{seq:2rmodel}
    H^{(2R)}_{I}& = \frac{\pi}{4} \sum_{i=1}^{L} \sigma^{z}_{i} \sigma^{z}_{i+2} + \sum_{i}\left(h_{i}-\frac{3 \pi}{4}\right)\sigma_{i}^{z},\nonumber \\
    H_{K}& = b\sum_{i=1}^{L} \, \sigma^{x}_{i}.
\end{align}
and the total Floquet operator is 
\begin{equation}\label{seq:2rangeu}
    U^{(2R)}_{\text{KFIM}}= e^{-iH_{K}}e^{-iH^{(2R)}_{I}}.
\end{equation}
We note that model given by Eq. \eqref{seq:2rmodel} for $|b|= \frac{\pi}{4}$ is dual unitary as shown in Ref. \cite{Pathak:2026jfn,Osipov:2026xyr}.
%~~~~~~~~~~~~~~~~~~~~~~~~~~~~~~~~~~~~~~~~~~~~~~~~~~~~~~~~~~~~~~~~~~~~~~
\subsection{Four-body case}\label{ssec:4bodycase}
We next consider the case of $p=4$, for which we have the following Hamiltonian

\begin{align}
    H^{(4B)}_{I}& = \sum_{i=1}^{L} J \sigma^{z}_{i}\sigma^{z}_{i+1} \sigma^{z}_{i+2}\sigma^{z}_{i+3} + \sum_{i}h_{i}\sigma_{i}^{z},\nonumber \\
    H_{K}& = \sum_{i=1}^{L} b\, \sigma^{x}_{i}.
\end{align}
We set $J=\frac{\pi}{4}$ and $b$ is arbitrary. And the total Floquet operator is 
\begin{equation}\label{seq:pbodyu}
    U^{(4B)}_{\text{KFIM}}= U_{K}U_{I}, \quad pB \equiv \text{p-body}.
\end{equation}

Focusing on the $U_{I}$ part of the total Floquet operator and using Lemma 1 we obtain the following decomposition
\begin{equation}\label{seq:decomp4}
 e^{-i\frac{\pi}{4}\sum \sz_{i}\sz_{i+1}\sz_{i+2}\sz_{i+3} - i \sum_{i=1}^{L}h_{i}\sz_{i} }= e^{-i\frac{\pi}{4}\sum \sz_{i}\sz_{i+1}\sz_{i+2}\sz_{i+3}}e^{-i(-\frac{\pi}{4})\sum_{i=1}^{L}(1-\sz_{i})(1-\sz_{i+1})(1-\sz_{i+2})(1-\sz_{i+3}) -i \sum_{i=1}^{L}h_{i}\sz_{i} }. 
\end{equation}
Using identity \eqref{suppsec:claim1} for $p=4$, we obtain the following
\begin{align}
  \prod_{k=1}^{4} (1-\sz_{i_k})&= 1- (\sz_{i_1}+\sz_{i_2}+\sz_{i_3}+\sz_{i_{4}}) + (\sz_{i_1}\sz_{i_2} +\sz_{i_1}\sz_{i_3}+\sz_{i_1}\sz_{i_4}+\sz_{i_2}\sz_{i_3}+\sz_{i_2}\sz_{i_4}+\sz_{i_3}\sz_{i_4}) \nonumber \\
  &- (
\sz_{i_{1}}\sz_{i_{2}}\sz_{i_{3}}+
\sz_{i_{1}}\sz_{i_{2}}\sz_{i_{4}}+
\sz_{i_{1}}\sz_{i_{3}}\sz_{i_{4}}+
\sz_{i_{2}}\sz_{i_{3}}\sz_{i_{4}})+ 
\sz_{i_{1}}\sz_{i_{2}}\sz_{i_{3}}\sz_{i_{4}}.   
\end{align}
Further, using identity  \eqref{suppsec:claim1} for $p=3$ recursively for every term, we obtain following for Eq. \eqref{seq:decomp4}
\begin{align}
&\exp\left[-i\frac{\pi}{4}\sum \sz_{i}\sz_{i+1}\sz_{i+2}\sz_{i+3} -i \sum_{i=1}^{L}h_{i}\sz_{i} \right] \nonumber \\
&= \exp\left[ -i \frac{\pi}{4}\left(3L -2\sum_{i=1}^{L} \sum_{k=0}^{3}\sz_{i+k} + \sum_{i=1}^{L} \sum_{k=0}^{2}\sz_{i+k}\sz_{i+k+1}+ \sum_{i=1}^{L} \sum_{k=0}^{1}\sz_{i+k}\sz_{i+k+2}+\sum_{i=1}^{L} \sz_{i}\sz_{i+3}\right)-i \sum_{i=1}^{L}h_{i}\sz_{i}\right], \nonumber \\
&= \exp\left[ -i \frac{\pi}{4}\left(3L -8\sum_{i=1}^{L}\sz_{i} + 3\sum_{i=1}^{L} \sz_{i}\sz_{i+1}+ 2\sum_{i=1}^{L} \sz_{i}\sz_{i+2}+\sum_{i=1}^{L} \sz_{i}\sz_{i+3}\right)-i \sum_{i=1}^{L}h_{i}\sz_{i}\right],\nonumber \\
&\simeq \exp\left[ -i \left(\frac{3\pi}{4}\sum_{i=1}^{L} \sz_{i}\sz_{i+1}+\frac{\pi}{4}\sum_{i=1}^{L} \sz_{i}\sz_{i+3}\right)-i \sum_{i=1}^{L}h_{i}\sz_{i}\right],
\end{align}
where in the last line $\simeq$ implies that the expressions are equivalent up to a global phase. Thus we have explicitly mapped the discrete time evolution of $4-$body model to the discrete time evolution of a model with $2-$body interactions having maximal range of 3. The model is explicitly given as follows 
\begin{align}\label{seq:3rmodel}
    H^{(3R)}_{I}&= \frac{3\pi}{4} \sum_{i=1}^{L} \sigma^{z}_{i}\sigma^{z}_{i+1} +\frac{\pi}{4} \sum_{i=1}^{L} \sigma^{z}_{i} \sigma^{z}_{i+3} + \sum_{i} h_{i} \sigma_{i}^{z},\nonumber \\
    H_{K}& = b\sum_{i=1}^{L} \, \sigma^{x}_{i}.
\end{align}
and the total Floquet operator is 
\begin{equation}\label{seq:2rangeu}
    U^{(3R)}_{\text{KFIM}}= e^{-iH_{K}}e^{-iH^{(3R)}_{I}}.
\end{equation}
The mapping is true for arbitrary $b$. For $|b|= \frac{\pi}{4}$ the model can be shown to be dual unitary, as shown in Ref. \cite{Pathak:2026jfn,Osipov:2026xyr}.
%~~~~~~~~~~~~~~~~~~~~~~~~~~~~~~~~~~~~~~~~~~~~~~~~~~~~~~~~~~~~~~~~~~~~~~
\subsection{$p$-body case}
To proceed with the general $p$- body case we first show that following identity holds
\begin{claim}
For $p >3$ we have the following
\begin{equation}
   \exp\left[-i \frac{\pi}{4} \sz_{i_{1}}\sz_{i_{2}}\cdots\sz_{i_{p}} \right]= \exp\left[-i \frac{\pi}{4}\left( \binom{p-1}{2}-  (p-2) \sum_{k=1}^{p} \sz_{i_{k}} +\sum_{1\leq k_1 < k_{2} \leq p} \sz_{i_{k_1}}\sz_{i_{k_2}}\right)\right].
\end{equation}
\end{claim}

\emph{Proof.} We proceed via induction. The base cases of $p=3$ and $p=4$ are already shown to be true in Sec. \eqref{ssec:3bodycase} and \eqref{ssec:4bodycase} respectively. Next, assume that the identity holds for $(p-1)$-body and $p$-body cases. 
Next, consider the case of $(p+1)$-body. Denoting $A= \sigma_{i_{1}}\sigma_{i_{2}}\cdots \sigma_{i_{p-1}}$ we have 
\begin{align}\label{seq:inducclaim2}
    \exp\left[-i \frac{\pi}{4} \sigma_{i_{1}}\sigma_{i_{2}}\cdots\sigma_{i_{p+1}} \right] &=\exp\left[-i \frac{\pi}{4} A\sigma_{i_{p}}\sigma_{i_{p+1}} \right], \nonumber \\
    &=\exp\left[-i \frac{\pi}{4}\left( 1  - (A+\sz_{i_{p}}+\sz_{i_{p+1}}) + A\sz_{i_{p}}+A\sz_{i_{p+1}} +\sz_{i_{p}}\sz_{i_{p+1}}\right)\right],
\end{align}
where we use results of three-body case in the last line. Next, we use the following expression 
\begin{align}
    \exp\left[-i \frac{\pi}{4}A \sigma_{i_{p}}\right]&= \exp\left[-i \frac{\pi}{4}\left( \binom{p-1}{2}-  (p-2) \left(\sum_{k=1}^{p-1} \sz_{i_{k}} + \sz_{i_{p}}\right) +\sum_{1\leq k_1 < k_{2} \leq p} \sz_{i_{k_1}}\sz_{i_{k_2}}\right)\right], \\
    \exp\left[-i \frac{\pi}{4}A \sigma_{i_{p+1}}\right]&= \exp\left[-i \frac{\pi}{4}\left( \binom{p-1}{2}-  (p-2) \left(\sum_{k=1}^{p-1} \sz_{i_{k}} + \sz_{i_{p+1}}\right) +\sum_{1\leq k_1 < k_{2} \leq p-1} \sz_{i_{k_1}}\sz_{i_{k_2}}+ \sum_{k=1}^{p-1}\sz_{i_{k}}\sigma_{i_{p+1}}\right)\right],\\
    \exp\left[-i \frac{\pi}{4}A \right]&= \exp\left[-i \frac{\pi}{4}\left( \binom{p-2}{2}-  (p-3) \left(\sum_{k=1}^{p-1} \sz_{i_{k}} \right) +\sum_{1\leq k_1 < k_{2} \leq p-1} \sz_{i_{k_1}}\sz_{i_{k_2}}\right)\right],
\end{align}
and obtain the following expression for Eq. \eqref{seq:inducclaim2}
\begin{align}
  \exp\left[-i \frac{\pi}{4} A\sigma_{i_{p}}\sigma_{i_{p+1}} \right]
    &=\exp\left[-i \frac{\pi}{4}\left( 1  - (A+\sz_{i_{p}}+\sz_{i_{p+1}}) + A\sz_{i_{p}}+A\sz_{i_{p+1}} +\sz_{i_{p}}\sz_{i_{p+1}}\right)\right], \nonumber \\
    &= \exp\left[-i \frac{\pi}{4}\left( \binom{p}{2} -(p-1)\sum_{k=1}^{p+1}\sz_{i_{k}} + \sum_{1\leq k_1 < k_{2} \leq p+1} \sz_{i_{k_1}}\sz_{i_{k_2}}\right)\right],
\end{align}
which is the result for $(p+1)-$body case. This concludes the proof. 

For a chain of length $L$ we thus have 
\begin{equation}
   \exp\left[-i \frac{\pi}{4} \sum_{i=1}^{L}\sz_{i}\sz_{i+{1}}\cdots\sz_{i+{p-1}} \right]= \exp\left[-i \frac{\pi}{4} \sum_{i=1}^{L}\left( \binom{p-1}{2}-  (p-2) \sum_{k=0}^{p-1} \sz_{i+{k}} +\sum_{0\leq k_1 < k_{2} \leq p-1} \sz_{i+{k_1}}\sz_{i+{k_2}}\right)\right].
\end{equation}
Using the fact that we have periodic boundary condition we have following expressions
\begin{align}
    &\sum_{i=1}^{L}\sum_{k=0}^{p-1} \sz_{i+k}= p \sum_{i=1}^{L}\sz_{i}, \\
    &\sum_{i=1}^{L}\sum_{0\leq k_1 < k_{2} \leq p-1} \sz_{i+{k_1}}\sz_{i+{k_2}}= \sum_{i=1}^{L}\sum_{r=1}^{p-1} (p-r)\sz_{i}\sz_{i+r}.
\end{align}

Using these we finally obtain 
\begin{equation}
   \exp\left[-i \frac{\pi}{4} \sum_{i=1}^{L}\sz_{i}\sz_{i+{1}}\cdots\sz_{i+{p-1}} \right]= \exp\left[-i \frac{\pi}{4} \left( L \binom{p-1}{2}-  p(p-2) \sum_{i=1}^{L} \sz_{i} +\sum_{r=1}^{p-1}(p-r) \sum_{i=1}^{L}\sz_{i}\sz_{i+r}\right)\right].
\end{equation}
Thus we have mapped the Floquet evolution operator of $p$- body model to a Floquet evolution operator of a model with 2-body interactions having maximal range of $p-1$. We can explicitly read off the corresponding Hamiltonian (both the kicked and the interaction part) which is as follows
\begin{align}\label{seq:3rmodel}
 H^{((p-1)R)}_{\text{KFIM}}&= H_{I} + H_{K} \sum_{\tau= -\infty}^{\infty} \delta(t- \tau),
 \end{align}
where  $$H^{((p-1)R)}_{I} = \frac{\pi}{4} \sum_{r=1}^{p-1}(p-r) \sum_{i=1}^{L}\sz_{i}\sz_{i+r} + \sum_{i} \left(h_{i} -p(p-2)\frac{\pi}{4}\right) \sigma_{i}^{z}, \quad 
    H_{K} = b\sum_{i=1}^{L} \, \sigma^{x}_{i}$$

This concludes the proof.

\section{Proof of Corollary 1} \label{ssec:cor1proof}
\begin{rcoro}
    For odd $p$, with periodic boundary conditions, $J= \frac{(2n+1) \pi}{4}$, we have the following two cases:
    \begin{itemize}
    \item For even $L$ the model is mapped, up to a global phase, to two decoupled chains each of which has maximum range $\frac{p-1}{2}$ interactions. 
        \item For odd $L$ the model is mapped, up to a global phase, to a single chain with maximum range $\frac{p-1}{2}$ interactions. 
    \end{itemize}
\end{rcoro}
\emph{Proof.}  
We first note that $p-r$ gives 1 for $r= p-1$. Thus there are always range $p-1$ interactions in the mapped model. The lattice sites can then be arranged into sublattices each of which contains sites denoted by $i$, such that $\bmod(p-1, i) = k; k \in \{0,1,\cdots,p-2\}$. Further note that there are two kind of sublattices: 
\begin{itemize}
    \item Type $A$: $k \in \{0,2,\cdots,p-2\}$ which contains only even lattice sites. 
    \item Type $B$ : $k \in \{1,3,\cdots,p-3\}$ which contains only odd lattice sites. 
\end{itemize} 

Next, we note that $(p-r)$ is even when $r$ is odd and thus it contributes a trivial global phase to the total Floquet operator. Thus the only remaining non-trivial couplings have only even ranges as follows: 
$$ 2, 4, \cdots p-3, p-1.$$

We now explicitly look at the two cases of even and odd $L$ as follows.

\subsection{Even $L$}
Start by arranging the $L$ lattice sites into sublattices where each of the sublattices consist of sites denotes by $i$ such that $\bmod(p-1, i) = k; k \in \{0,1,\cdots,p-2\}$, together. Note that range $p-1$ interaction couples the sites within each of the sublattices together, while the sublattices are at present decoupled from each other. Next, using the fact that each of the sublattices only contains only even or only odd site and that the range of non-trivial couplings is even we come to the conclusion that sublattices of type $A$ are coupled to each other and type $B$ are coupled to each other while none of the sublattices belonging to type $A$ and $B$ are coupled to each other. Furthermore, for each type of the sublattices $A$ and $B$ we can do the relabeling of the sites as follows
$$
i \mapsto \frac{i}{2}\, \text{for even sites},\nonumber \\
  i \mapsto \frac{i+1}{2} \, \text{for odd sites.} 
$$
Thus we obtain two decoupled chains each consisting of range $\frac{p-1}{2}$ interactions. 

\subsection{Odd $L$}
For the case of odd $L$ we can proceed the same as the even $L$ case. However, we now note that we have one isolated site, $L$, which stands out. Since we are working with periodic boundary condition we get following interactions due to it
$$ (L, p-1), (L-1, p-2), (L-2, p-2), \cdots (L -(p-2), 1)$$

which results in an interaction between sublattices of type $A$ and $B$ thus given a single coupled chain with range $\frac{p-1}{2}$ interactions. 

This concludes the proof. 
\end{document}